\documentclass[paper, 10 pt, journal]{ieeeconf}

\IEEEoverridecommandlockouts                              

\usepackage{amsmath}
\usepackage{amssymb}
\usepackage{amsfonts}
\usepackage{graphicx}
\usepackage{enumerate}
\usepackage{mathrsfs}
\usepackage{float}
\usepackage{cite}
\usepackage{tikz}
\usetikzlibrary{matrix,arrows,decorations.pathmorphing}
\usepackage{verbatim}
\usepackage{mathtools}
\usepackage{caption}
\usepackage{subcaption}
\usepackage{soul}
\setstcolor{magenta}
\usetikzlibrary{arrows}
\usepackage{tabularx}

\newtheorem{defi}{Definition}

\graphicspath{{./figs/}}

\title{\LARGE \bf
On Tensor-based Polynomial Hamiltonian Systems }

\author{Shaoxuan Cui$^{1,4}$, Guofeng Zhang$^{2}$, Hildeberto Jardón-Kojakhmetov$^{1}$ and Ming Cao$^{3}$ 
\thanks{$^{1}$ S. Cui, and H. Jard\'on-Kojakhmetov are with the Bernoulli Institute for Mathematics, Computer Science and Artificial Intelligence, University of Groningen, Groningen, 9747 AG Netherlands {\tt\small \{s.cui, h.jardon.kojakhmetov\}@rug.nl}}
\thanks{$^{2}$ G. Zhang is with the Department of Applied Mathematics, The Hong Kong Polytechnic University, Kowloon 999077, Hong Kong, China; Research Institute for Quantum Technology, The Hong Kong Polytechnic University,  Hong Kong, China and The Hong Kong Polytechnic University Shenzhen Research Institute, Shenzhen, Guang Dong 518057, China {\tt\small guofeng.zhang@polyu.edu.hk}}
\thanks{$^{3}$ M. Cao is with the Engineering and Technology institute Groningen, University of Groningen, Groningen, 9747 AG Netherlands {\tt\small m.cao@rug.nl}}
\thanks{$^{4}$ S.Cui was supported by China Scholarship Council. The work of M. Cao was supported in part by the Netherlands Organization for Scientific Research (NWO-Vici-19902). G. Zhang was supported by the Guangdong Provincial Quantum Science Strategic InitiativeNo. (GDZX2200001) and the National Natural Science Foundation of China (No. 62173288).}
}

\begin{document}

\maketitle

\thispagestyle{empty}
\pagestyle{empty}

\newtheorem{remark}{Remark}
\newtheorem{lemma}{Lemma}
\newtheorem{thm}{Theorem}
\newtheorem{example}{Example}
\newtheorem{definition}{Definition}
\newtheorem{prop}{Proposition}

\begin{abstract}
It is known that a linear system with a system matrix $A$ constitutes a Hamiltonian system with a quadratic Hamiltonian if and only if $A$ is a Hamiltonian matrix. This provides a straightforward method to verify whether a linear system is Hamiltonian or whether a given Hamiltonian function corresponds to a linear system. These techniques fundamentally rely on the properties of Hamiltonian matrices. Building on recent advances in tensor algebra, this paper generalizes such results to a broad class of polynomial systems. As the systems of interest can be naturally represented in tensor forms, we name them "tensor-based polynomial systems". Our main contribution is that we formally define Hamiltonian cubical tensors and characterize their properties. Crucially, we demonstrate that a tensor-based polynomial system is a Hamiltonian system with a polynomial Hamiltonian if and only if all associated system tensors are Hamiltonian cubical tensors—a direct parallel to the linear case. Additionally, we establish a computationally tractable stability criterion for tensor-based polynomial Hamiltonian systems. Finally, we validate all theoretical results through numerical examples and provide a further intuitive discussion.

\end{abstract}

\section{Introduction}

A Hamiltonian system constitutes a class of dynamical systems whose equations of motion can be expressed via a scalar function known as the Hamiltonian \cite{franccoise2006encyclopedia}. Originating from the analysis of physical systems, this mathematical framework provides a natural formalism for studying conservative mechanical systems \cite{easton1993introduction}. Canonical examples range from elementary systems like the harmonic oscillator (where the total mechanical energy remains conserved) to complex phenomena exemplified by the $N$-body problem \cite{easton1993introduction}. Modern theoretical developments have extended the Hamiltonian paradigm beyond classical mechanics, finding applications in diverse domains including ecological modeling \cite{plank1995hamiltonian,fernandes1995hamiltonian} and economic systems analysis \cite{scalia2020ecology,tarasyev2010construction}.
Notably, Hamiltonian systems exhibit profound connections with optimal control theory, as solutions to optimal control problems inherently follow Hamiltonian trajectories \cite{agrachev2008hamiltonian}. The control-theoretic perspective further considers Port-Hamiltonian systems, incorporating external inputs \cite{van2014port,van2006port,rashad2020twenty} -- a formulation particularly relevant in modern energy-based control methodologies.


The linear Hamiltonian system is the most fundamental type of Hamiltonian systems. It is well-established that a linear Hamiltonian system corresponds to a quadratic form of the Hamiltonian function \cite{willems1998quadratic,easton1993introduction}. By leveraging the properties of Hamiltonian matrices and the linear systems theory, the behavior of such systems has been extensively studied \cite{easton1993introduction}. Moreover, the properties of Hamiltonian matrices are widely utilized in the field of optimal control \cite{zhou1996robust}, including $H_2$
  control \cite{kwakernaak2000h2} and $H_\infty$
  control \cite{van1991state}. These properties also play a significant role in quantum linear systems \cite{zhang2019structural} and quantum Hamiltonian systems \cite{gough2021randomized}.

In contrast, the behavior of general nonlinear Hamiltonian systems has not been fully explored. A natural question arises: how can the well-understood case of linear Hamiltonian systems be generalized to uncover certain classes of nonlinear Hamiltonian systems?

The behavior of a linear system is closely related to its associated system matrix. With advancements in tensor algebra, it has recently been demonstrated that polynomial systems, where the power is a positive integer, can be represented in tensor forms. This representation allows for further analysis of the system using tensor properties \cite{chen2021stability,chen2021controllability,chen2022explicit,cui2024metzler,10540364}. However, most of these results \cite{chen2021stability,chen2021controllability,chen2022explicit,cui2024metzler,10540364} are limited to the case of asymptotic stability. For autonomous Hamiltonian systems, an equilibrium is at most Lyapunov stable \cite{easton1993introduction}. In this work, we aim to utilize tensor properties to study polynomial Hamiltonian systems.

The main contributions of this paper are as follows. First, we establish a necessary and sufficient condition for a polynomial system to be Hamiltonian with a polynomial Hamiltonian, using tensor algebra. Both the theoretical framework and the tensor techniques introduced are novel to the field. Second, we provide a stability criterion for each equilibrium of a tensor-based polynomial Hamiltonian system, which can be easily verified through tensor operations. Finally, we illustrate our theoretical results with application examples.

The paper is organized as follows: In Section \ref{sec:2}, we introduce preliminary concepts related to tensors, which are essential for deriving our main results. In Section \ref{sec:3}, we review the theory of linear Hamiltonian systems. In Section \ref{sec:4}, we present our main results on tensor-based polynomial Hamiltonian systems. In Section \ref{sec:5}, we demonstrate our findings with examples and discuss potential applications. Finally, in the last section, we conclude the paper and provide further discussion.

\emph{Notation:} $\mathbb{R}$ denotes the set of  real numbers. For a matrix $M \in \mathbb{R}^{n \times r}$ and a vector $a \in \mathbb{R}^n$, $M_{ij}$ and $a_{i}$ denote the element in the $i$th row and $j$th column and the $i$th entry, respectively. 
The symbol $\textbf{det}(R)$ denotes the determinant of a matrix or a tensor $R$ \cite{qi2005eigenvalues}. The symbol $\textbf{tr}(R)$ denotes the trace of a matrix or a tensor $R$ \cite{qi2005eigenvalues}. In the following section, we introduce the preliminaries on tensors and some further notations regarding tensors. 

\section{Tensor preliminaries}\label{sec:2}
A tensor $T\in\mathbb{R}^{n_1\times n_2 \times \cdots \times n_k}$ is a multidimensional array, where the order is defined as the number of its dimension $k$ and each dimension $n_i$, $i=1,\cdots,k,$ is called a mode of the tensor. A tensor is a cubical tensor if every mode has the same size, i.e., $T\in\mathbb{R}^{n\times n \times \cdots \times n}$. We further write a $k$-th order $n$ dimensional cubical tensor as $T\in\mathbb{R}^{n\times n \times \cdots \times n}=\mathbb{R}^{[k,n]}$. A cubical tensor $T$ is a supersymmetric tensor if $T_{j_1 j_2 \ldots j_k}$ is invariant under any permutation of the indices. For the rest of the paper, a tensor always refers to a cubical tensor.

For a tensor $A \in \mathbb{R}^{[k,n]}$ and a vector $x\in \mathbb{R}^n $, the tensor-vector product $Ax$ is a $(k-1)$-th order $n$ dimensional tensor (i.e., $Ax\in \mathbb{R}^{[k-1,n]}$), whose component is given by $$(Ax)_{i_1i_2\ldots i_{k-1}}=\sum_{i_k=1}^{n}A_{i_1\ldots i_{k-1}i_k}x_{i_k}.$$ 
Similarly, by repeatedly using the above operation, the component of the tensor-vector product to the power $k-2$, $Ax^{k-2}\in \mathbb{R}^{[2,n]}$ is defined as $$(Ax^{k-2})_{i_1i_2}=\sum_{i_3,\ldots, i_k=1}^{n}A_{i_1i_2i_3\ldots i_k}x_{i_3}\ldots x_{i_k}.$$ The tensor-vector product to the power $k-1$, $Ax^{k-1}\in \mathbb{R}^{[1,n]}$ yields a vector, whose $i$-th component is given by $$(Ax^{k-1})_{i_1}=\sum_{i_2,\ldots,i_k=1}^{n}A_{i_1i_2\ldots i_k}x_{i_2}\ldots x_{i_k}.$$
 The tensor-vector product to the power $k:Ax^{k}$ is defined as a scalar $$Ax^{k}=x^{\top}(Ax^{k-1})=\sum_{i_1,\ldots,i_k=1}^{n}A_{i_1i_2\ldots i_k}x_{i_1}\ldots x_{i_k}.$$ For a matrix $R\in\mathbb{R}^{n\times n}$ and a tensor $A \in \mathbb{R}^{[k,n]}$, the matrix-tensor product $RA \in \mathbb{R}^{[k,n]}$ is as follows: $$(RA)_{i_1i_2\ldots i_k}=\sum_{j=1}^{n}R_{i_1 j}A_{ji_2\ldots i_k}.$$

In addition, we have the following result.

\begin{lemma}\cite{chen2022explicit}\label{lemma_1}
Given a one dimensional homogenous polynomial function $f(x)$: $\mathbb{R}^n\rightarrow\mathbb{R}.$ The function $f(x)$ can be uniquely represented by the tensor form $Ax^m$, $m$ is the order of $A$, $A$ is a super-symmetric cubical tensor.
\end{lemma}

In the following, we prove the following associative law for matrices and tensors.

\begin{lemma}\label{lemma:asso}
    If $A \in \mathbb{R}^{[k,n]}$ is a tensor and $B,C \in \mathbb{R}^{n \times n}$ are matrices, $BCA=B(CA)$.
\end{lemma}

\begin{proof}
    We have that $(CA)_{i_1i_2\ldots i_k}=\sum_{j=1}^{n}C_{i_1j}A_{ji_2\ldots i_k}.$
    It follows that $\left(B(CA)\right)_{i_1i_2\ldots i_k}=\sum_{j=1}^{n}B_{i_1 j}(CA)_{ji_2\ldots i_k}= \sum_{j=1}^{n}B_{i_1 j}\left(\sum_{t=1}^{n}C_{jt}A_{ti_2\ldots i_k}\right).$

    On the other hand, $(BC)_{i_1 j}=\sum_j B_{i_1 t} C_{tj}$. Then, $(BCA)_{i_1\cdots i_k}=\sum_{j=1}^{n}(BC)_{i_1j}A_{ji_2\ldots i_k}=\sum_j \left(\sum_t B_{i_1 t} C_{tj}\right) A_{ji_2\ldots i_k}.$

    Thus, $BCA=B(CA)$.
\end{proof}

In addition, we have the following result.
\begin{lemma}\label{lemma:2}
    If $A \in \mathbb{R}^{[k,n]}$ is a tensor and $I\in \mathbb{R}^{n \times n}$ is an identity matrix, then $IA=A$.
\end{lemma}

\begin{proof}
    It is easy to see that $(IA)_{i_1i_2\ldots i_k}=I_{i_1 i_1}A_{i_1i_2\ldots i_k}= A_{i_1i_2\ldots i_k}.$
\end{proof}
 
Furthermore, we have the following notation regarding a vector $x\in \mathbb{R}^n$:
\begin{equation*}
    \left(x^{[{k}-1]}\right)_i  =x_i^{{k}-1}.
\end{equation*}

Recall that $A x^{{k}-1} \text{ and }  x^{[{k}-1]}$ are both vectors. For a tensor $A \in \mathbb{R}^{[k,n]}$, suppose the following homogeneous polynomial equation:
\begin{equation}\label{eq:eigenproblem}
A x^{{k}-1}=\lambda x^{[{k}-1]},
\end{equation}
has the solution consisting of a real number $\lambda$ and a nonzero real vector $x$ that are solutions of \eqref{eq:eigenproblem}. Then $\lambda$ is called an H-eigenvalue of $A$ and $x$ is the H-eigenvector of $A$ associated with $\lambda$ \cite{qi2005eigenvalues,zhang2014m,chang2013survey}. If $\lambda$ and $x$ are both complex, then they are called eigenvalues and eigenvectors of $A$ \cite{qi2005eigenvalues}.

The definition of a positive definite tensor originates from \cite{qi2005eigenvalues}. When ${A}\in \mathbb{R}^{[k,n]}$ is supersymmetric and $k$ is even, we say that ${A}$ is positive (semi-)definite, if ${A} {x}^k>(\geq)0$ for all ${x} \in \mathbb{R}^n$ and ${x} \neq \mathbf{0}$. A cubical supersymmetric tensor is positive (semi-)definite if and only if all its H-eigenvalues are positive (non-negative) \cite{qi2005eigenvalues}.

Next, we introduce the transpose of a tensor according to \cite{pandey2024linear,pan2014tensor,solo2010multidimensional}. Let $S=\{1,2,3, \cdots, k\}$ and $\sigma$ be a permutation of $S$. The permutation $\sigma$ can be represented as follows, with $\sigma(i)=w_i$
$$
\sigma=\left(\begin{array}{ccccc}
1 & 2 & 3 & \cdots & k \\
w_1 & w_2 & w_3 & \cdots & w_k
\end{array}\right)
$$
where $\left\{w_1, w_2, w_3, \cdots, w_k\right\}=\{1,2,3, \cdots, k\}=S$. Furthermore, two permutations can be composed. For example, $\sigma$ and $\tau$ are two permutations and their composition $\tau \sigma(i)=\tau(\sigma(i))$. The set consisting of all permutations of any given set $S$, together with the composition of functions is the symmetry group of $S$, denoted by $S_k$. Obviously, there are $k!$ permutations for the set $S$. 

\begin{defi}[\cite{pan2014tensor}]
    Let ${X}\in \mathbb{R}^{[k,n]}$ be an $n$-th order tensor. The tensor ${Y}$ is called tensor transpose of ${X}$ associated with $\sigma$, if entries ${Y}\left(i_{\sigma(1)}, i_{\sigma(2)}, \ldots, i_{\sigma(k)}\right)={X}\left(i_1, i_2, \ldots, i_k\right)$, where $\sigma$ is an element of $S_k$ but not an identity permutation. ${Y}$ is denoted by the ${X}^{\top_\sigma}$.
\end{defi}

For a matrix, which contains only two indices, the only non-identity permutation is \( \sigma \) such that \( \sigma(1) = 2 \) and \( \sigma(2) = 1 \). Consequently, the transpose of a matrix \( A \) can be unambiguously written as \( A^\top \). We know that a matrix \( A \) is symmetric if \( A^\top = A \). Similarly, the concept of a supersymmetric tensor, as mentioned earlier, can also be defined using the transpose operation applied to tensors.

\begin{defi}[\cite{pan2014tensor}]
    A cubical tensor ${X}$ is called a supersymmetric tensor, if ${X}={X}^{T _\sigma}$, for all $\sigma \in S_k$, where $k$ is the order of ${X}$.
\end{defi}

Finally, the following result is presented and used in for example \cite{qi2005eigenvalues} but the proof is omitted in \cite{qi2005eigenvalues}. To help the readers who are not familiar with tensors, we give a detailed proof.
\begin{lemma}
    If $B$ is a supersymmetric tensor of order $k$, then $\frac{\partial (Bx^k)}{\partial x}=kBx^{k-1}$, $\frac{\partial (Bx^{k-1})}{\partial x}=(k-1)Bx^{k-2}$.
\end{lemma}

\begin{proof}
    We know that $Bx^{k}=\sum_{i_1,\ldots,i_k=1}^{n}B_{i_1i_2\ldots i_k}x_{i_1}\ldots x_{i_k}.$ Then, $\frac{\partial (Bx^k)}{\partial x_j} = \sum_{i_1=j} B_{i_1i_2\ldots i_k}x_{i_2}\ldots x_{i_k}+\cdots+ \sum_{i_k=j} B_{i_1i_2\ldots i_k}x_{i_2}\ldots x_{i_{k-1}}$. As $B$ is supersymmetric, $\frac{\partial (Bx^k)}{\partial x}=kBx^{k-1}$.

    The equation $\frac{\partial (Bx^{k-1})}{\partial x}=(k-1)Bx^{k-2}$ can be derived similarly by checking its componentwise representation.
\end{proof}

Let the shorthand $\sigma_{ij}$ be the permutation, where only $i,j$ are transposed and all other indices remain the same.
It follows that if the tensor $B$ is not supersymmetric, then $\frac{\partial (Bx^k)}{\partial x}=Bx^{k-1}+\sum^k_{j=2}B^{\top_{\sigma_{1j}}}x^{k-1}$. Naturally, if $B$ is a matrix, then $\frac{\partial (x^\top Bx)}{\partial x}=Bx+B^{\top}x$.

\section{Linear Hamiltonian system}\label{sec:3}
In this section, we make a recap of the linear Hamiltonian system.
A linear Hamiltonian system is a system of $2n$ ordinary differential equations:
\begin{equation}\label{eq:lin}
\dot{z}=J \frac{\partial H}{\partial z}=J B z=A z,
\end{equation}
where $J$ is a symplectic matrix and $A=JB.$. Typically $J$ is the block matrix
$$
J=\left[\begin{array}{cc}
0 & I_n \\
-I_n & 0
\end{array}\right];
$$
where $I_n$ is the $n \times n$ identity matrix. The matrix $J$ has determinant $+1$ and its inverse is $J^{-1}=J^\top=-J$.
The Hamiltonian which corresponds to the system \eqref{eq:lin} is $H=\frac{1}{2} z^\top B z=Bz^2$ and it is a quadratic form in $z$.

A matrix $A$ is called a Hamiltonian matrix, if $$
A^T J+J A=0.
$$

It is straightforward to see that, if $A$ is a Hamiltonian matrix, then the linear system $\dot{z}=A z$ is a Hamiltonian system with a quadratic Hamiltonian.

\begin{prop}[Theorem 2.1.1 \cite{easton1993introduction}]\label{prop:matrix}
The following are equivalent for a matrix $A$:
\begin{itemize}
    \item[i)] $A$ is a Hamiltonian matrix,
    \item[ii)] $A=J R$ where $R$ is symmetric,
    \item[iii)] $J A$ is symmetric.
\end{itemize}
\end{prop}

This proposition provides a straightforward way to identify a linear Hamiltonian system. We will generalize the results to a class of nonlinear polynomial system in the following section.

\section{Tensor-based Polynomial Hamiltonian System}\label{sec:4}
In this paper, we consider an arbitrary polynomial system in the tensor form:
\begin{equation}\label{eq:model1}
    \Dot{x}=A_k x^{k-1}+ A_{k-1} x^{k-2}+\cdots+ A_2 x, A_j\in \mathbb{R}^{[j,n]}.
\end{equation}
We observe that any polynomial system with integer power can be expressed in the standard form of \eqref{eq:model1}. A central focus of this paper concerns the Hamiltonian characterization of \eqref{eq:model1}: specifically determining under what conditions this formulation constitutes a Hamiltonian system. This fundamental question will be rigorously examined later in this paper.


Conversely, consider the polynomial Hamiltonian function:
\begin{equation}\label{eq:hami}
    H(x) = B_k x^{k} + B_{k-1}x^{k-1} + \cdots + B_2 x^2, B_j\in \mathbb{R}^{[j,n]}
\end{equation}
where all exponents are positive integers. By Lemma~\ref{lemma_1}, we may assume the coefficient tensors $B_k,\ldots,B_2$ are supersymmetric without loss of generality, given the scalar nature of $H(x)$. 
This formulation naturally raises the inverse problem: What class of polynomial Hamiltonian systems admits a representation through \eqref{eq:hami}? The canonical solution emerges through the symplectic structure:
{\small\begin{equation}\label{eq:model2}
    \Dot{x}=J \frac{\partial H}{\partial x}=k J B_k x^{k-1}+ (k-1)J B_{k-1} x^{k-2}+\cdots+ 2J B_2 x.
\end{equation}}

We observe that \eqref{eq:model2} takes the form of \eqref{eq:model1}. With this in mind, we now address the question: When is \eqref{eq:model1} a Hamiltonian system? To answer this, we first introduce the concept of a Hamiltonian cubical tensor.

\begin{defi}\label{def:ht}
    A cubical tensor $A$ is a Hamiltonian cubical tensor if $(J^\top A)^{\top_{\sigma}}+JA=0$ for any permutation $\sigma$.
\end{defi}

\begin{remark}
   In \cite{wang2024algebraic}, an alternative definition of a Hamiltonian tensor has been proposed, designed to ensure that Hamiltonian tensors exhibit algebraic properties similar to those of Hamiltonian matrices. Such a definition is particularly useful for studying the properties of a class of multilinear systems \cite{wang2024algebraic}. However, to the best of our knowledge, this definition \cite{wang2024algebraic} is not directly connected to Hamiltonian systems. In contrast, our Definition \ref{def:ht} is intrinsically linked to Hamiltonian systems. Moreover, we will demonstrate later that the Hamiltonian cubical tensor, as defined in Definition \ref{def:ht}, also shares some properties analogous to those of a Hamiltonian matrix.
\end{remark}

Notice that if $A$ is a square matrix, Definition \ref{def:ht} reduces to the standard definition of a Hamiltonian matrix. However, for tensors, the term $(J^\top A)^{\top_{\sigma}}$ generally cannot be simplified further. Analogous to the properties of Hamiltonian matrices outlined in Proposition \ref{prop:matrix}, a Hamiltonian cubical tensor exhibits the following equivalence properties, which closely mirror those of the matrix case discussed in Proposition \ref{prop:matrix}.

\begin{prop}\label{prop:hamiten}
    The following statements are equivalent:
    \begin{itemize}
        \item[i).] $A$ is a Hamiltonian cubical tensor.
        \item[ii).] $A=JR$ where $R$ is supersymmetric.
        \item[iii).] $JA$ is supersymmetric.
    \end{itemize}
\end{prop}

\begin{proof}
    Firstly, we show i) implies ii). Since $A$ is a Hamiltonian cubical tensor, it holds $(J^\top A)^{\top_{\sigma}}+JA=0$ for any permutation $\sigma$. Let $A=JR$, which means $R=(J^{-1})JR=IR=J^{-1} A=J^\top A$. Here, both the associative law illustrated in Lemma \ref{lemma:asso} and the result of Lemma \ref{lemma:2} are used. This further shows $R^{\top_{\sigma}}=-JA$ for any permutation $\sigma$. Then, $R$ is fixed under any permutation. $R$ is thus supersymmetric.

    Secondly, we show ii) implies iii). Let $A=JR$, $R$ is supersymmetric. Then, $JA=J^2R=-IR$, which is supersymmetric.

    Thirdly, we show iii) implies ii). If $JA$ is supersymmetric, $A=J^\top JA=-JJA=J(-JA)=JR$, where $R=-JA$ is supersymmetric.

    Finally, we show ii) implies i). We have $R=J^\top A=R^{\top_{\sigma}}=-JA$ for any $\sigma$. This further implies $(J^\top A)^{\top_{\sigma}}+JA=0$ for any index $\sigma$.
\end{proof}

Furthermore, a cubical Hamiltonian tensor has the following properties regarding the eigenvalues.

\begin{thm}
    If $A=JR$ is a cubical Hamiltonian tensor where $R$ is a supersymmetric tensor guaranteed by the Proposition \ref{prop:hamiten}, the product of all the eigenvalues of $A$ is $\textbf{det}(R)$ while the sum of all eigenvalues of $A$ is zero. 
\end{thm}

\begin{proof}
    According to Theorem 1 in \cite{qi2005eigenvalues}, the product is equal to $\textbf{det}(R)$. According to the definition of higher order determinant \cite{qi2005eigenvalues}, $\textbf{det}(R)=\textbf{det}(A)$. 

    Again, according to Theorem 1 in \cite{qi2005eigenvalues}, the sum of all eigenvalues of $A\in \mathbb{R}^{[k,n]}$ is $(k-1)^{n-1}\textbf{tr} (A)$. The trace of a tensor \cite{qi2005eigenvalues} is the sum of all diagonal elements of a tensor. Since $A=JR$, it is easy to know that if $i_1\leq \frac{n}{2}$, $A_{i_1 jj\cdots j}=R_{j\cdots j}$; otherwise, $A_{i_1 jj\cdots j}=-R_{j\cdots j}$.
    Thus, $\textbf{tr} (A)=0$.
\end{proof}

Next, we are able to present the first main result of this paper.

\begin{thm}
    Consider the system in the tensor form \eqref{eq:model1}. The system \eqref{eq:model1} is a Hamiltonian system with the Hamiltonian \eqref{eq:hami} if and only if $A_k,\cdots, A_2$ are all cubical Hamiltonian tensors.
\end{thm}

\begin{proof}
    If $A_k,\cdots, A_2$ are all cubical Hamiltonian tensors, then by Proposition \ref{prop:hamiten}, $A_i=JR_i$. Let $B_i=\frac{1}{i-1}R_i$. It is straightforward to know \eqref{eq:model1} is in the form of \eqref{eq:model2}. Thus, \eqref{eq:model1} is a Hamiltonian system with  Hamiltonian \eqref{eq:hami}.

    Conversely, if the Hamiltonian $H$ is in the form of \eqref{eq:hami}, it results in the system of \eqref{eq:model2}, which is naturally in the form of \eqref{eq:model1}. We observe from \eqref{eq:model2} that $A_i= J(i-1)B_i$. Since $(i-1)B_i$ is supersymmetric, $A_i$ is a cubical Hamiltonian tensor.
\end{proof}

We can see that the origin is an equilibrium of \eqref{eq:model1}. Then, we continue to look at the stability of the origin. We need to utilize the famous Theorem of Dirichlet.

\begin{lemma}[Dirichlet,\cite{easton1993introduction}]
    If $\xi$ is a strict local minimum or maximum of the Hamiltonian $H$, then $\xi$ is Lyapunov stable.
\end{lemma}

The criterion can be checked by calculating the eigenvalues of the Hessian matrix of $H$.

\begin{thm}\label{thm:origin}
    Consider the Hamiltonian system in the tensor form \eqref{eq:model2}. The origin is Lyapunov stable if the matrix $B_2$ is positive definite or negative definite.
\end{thm}

\begin{proof}
As all $B_i$ is supersymmetric,    we have 
    \begin{equation}
         \frac{\partial H}{\partial x}=k B_k x^{k-1}+ (k-1) B_{k-1} x^{k-2}+\cdots+ 2 B_2 x.
    \end{equation}
    In addition, 
    \begin{equation}\label{eq:hessian}
         \frac{\partial^2 H}{\partial x^2}=k(k-1) B_k x^{k-2}+ (k-1)(k-2) B_{k-1} x^{k-3}+\cdots+ 2 B_2.
    \end{equation}

    As for the origin, all higher order terms of $\frac{\partial^2 H}{\partial x^2}$ vanish, if $B_2$ is positive or negative definite, the origin is a strict local mimimum or maximum of $H$. Then, the origin is stable.
\end{proof}

For the other equilibria, the higher order term will play an active role in the stability.

\begin{thm}
    Consider the Hamiltonian system in the tensor form \eqref{eq:model2}. A non-zero equilibrium $x^*$ is stable if the matrix $k(k-1) B_k (x^*)^{k-2}+ (k-1)(k-2) B_{k-1} (x^*)^{k-3}+\cdots+ 2 B_2$ is positive definite or negative definite.
\end{thm}

\begin{proof}
    The proof is the same as in the proof of Theorem \ref{thm:origin}.
\end{proof}

In many real-world systems, the Hamiltonian function \eqref{prop:hamiten} is required to be bounded from below. We now investigate the case where the Hamiltonian \eqref{prop:hamiten} satisfies such a condition. A natural approach is to consider whether \eqref{prop:hamiten} is positive (semi-)definite. For instance, in many physical systems, the potential function is positive definite. However, verifying the positive definiteness of a nonlinear function is generally challenging. Leveraging tensor properties, we establish the following criteria to check for positive definiteness.

\begin{prop}
    The Hamiltonian function \eqref{prop:hamiten} is positive (semi-)definite if all tensors $B_i$ are positive (semi-)definite and there are no odd-order tensors, i.e., there are no odd $i$.
\end{prop}

\begin{proof}
    This is a direct consequence of the definition of a positive (semi-)definite tensor \cite{qi2005eigenvalues}. Furthermore, the sum of several positive (semi-)definite functions is positive (semi-)definite.
\end{proof}

To determine whether an even-order tensor is positive (semi-)definite, one just need to calculate the eigenvalues of the tensor. If all eigenvalues of the tensor are positive (non-negative), it is then positive (semi-)definite \cite{qi2005eigenvalues}.

\section{Applications and Examples}\label{sec:5}

In this section, we present several simple yet illustrative examples to demonstrate our main results and the effectiveness of the tensor-based approach.

Firstly, consider the system
\begin{equation}\label{eq:exa1}
    \begin{split}
        \Dot{x}_1&= x_1^2+2x_2,\\
        \Dot{x}_2&=-2x_1 x_2.
    \end{split}
\end{equation}
It is straightforward to verify that \eqref{eq:exa1} takes the form of \eqref{eq:model1}. The system tensors are given by \( A_2 = \begin{pmatrix} 0 & 2 \\ 0 & 0 \end{pmatrix} \) and \( A_3 \), a third-order tensor where \( (A_3)_{111} = 1 \), \( (A_3)_{212} = (A_3)_{221} = -1 \), and all other entries are zero. We then confirm that \( JA_2 \) is symmetric and \( JA_3 \) is supersymmetric. By Proposition \ref{prop:hamiten}, both tensors are Hamiltonian cubical tensors. Consequently, the system \eqref{eq:exa1} is a Hamiltonian system. Specifically, the Hamiltonian is given by the polynomial \( H = x_1^2 x_2 + x_2^2 \). This example illustrates that, to determine whether a polynomial system is a Hamiltonian system with a polynomial Hamiltonian, one should first rewrite the system in the tensor form \eqref{eq:model1} and then verify whether all \( JA_i \) are supersymmetric.

Another example is the anharmonic oscillator. The Hamiltonian equations of motion for an anharmonic oscillator with potential \( U = \frac{k x^2}{2} + \frac{b x^4}{4} \) can be derived as follows. First, compute the Lagrangian of the system, then determine the conjugate momentum. The Hamiltonian of the system is obtained as
$H = \frac{1}{2} m \dot{x}^2 + \frac{k x^2}{2} + \frac{b x^4}{4}.$
Substituting into the Hamiltonian equations, the partial derivatives are
$\frac{\partial H}{\partial p} = m \dot{x} \quad \text{and} \quad \frac{\partial H}{\partial x} = k x + b x^3.$
Recalling the definition of the conjugate momentum \( p = m \dot{x} \), the Hamiltonian equations of motion for the anharmonic oscillator are
$\frac{dx}{dt} = \frac{p}{m} \quad \text{and} \quad \frac{dp}{dt} = -k x - b x^3.$
This system is in the form of \eqref{eq:model1}. The system matrix is given by
$A_2 = \begin{pmatrix}
    0 & \frac{1}{m} \\
    -k & 0
\end{pmatrix},$
and the system tensor \( A_4 \) has the entry \( (A_4)_{2111} = -b \), with all other entries being zero. By verifying \( JA_2 \) and \( JA_4 \) are super-symmetric, it is straightforward to confirm that they are Hamiltonian tensors. 
This result is particularly useful when only the system dynamics are known, and one aims to determine whether the system is Hamiltonian and identify the associated Hamiltonian function.


In the case where the potential function is nonpolynomial but smooth, suppose the origin is an equilibrium (which is related to the reference point in the potential function). We can perform a Taylor expansion around the origin: $V(x)=V(0)+\Dot{V}(0) x+k_2 x^2+ k_3 x^3+ k_4 x^4+k_5 x^5+k_6 x^6+\cdots$. 
Since the origin is an equilibrium, \( \dot{V}(0) = 0 \). Additionally, as the potential must typically be lower bounded, the odd-order terms \( k_3 x^3, k_5 x^5, \dots \) must be absent. The remaining terms yield a result similar to the previous case, as the Hamiltonian takes the form
$H= \frac{1}{2} m \dot{x}^2+V(x)$ and \( V(x) \) can be approximated by its Taylor expansion around the origin. This implies that the anharmonic oscillator with a nonpolynomial but smooth potential function behaves similarly to a polynomial Hamiltonian system in the neighborhood of the origin.  

The Fermi–Pasta–Ulam–Tsingou problem \cite{ganapa2023quasiperiodicity,fermi1955studies,gallavotti2007fermi} is another example of a polynomial Hamiltonian system with a polynomial Hamiltonian. Fermi, Pasta, Ulam, and Tsingou simulated the vibrating string of the following form:
 \begin{equation}
 \begin{split}
     m \ddot{x}_j&=k\left(x_{j+1}+x_{j-1}-2 x_j\right)
     \left[1+\alpha\left(x_{j+1}-x_{j-1}\right)\right]\\
     &=f(x), 1\leq j\leq n.
 \end{split}
\end{equation}
The whole formula is nothing but a Newton's second law for the $j$-th particle. The first term $k\left(x_{j+1}+x_{j-1}-2 x_j\right)$ is the usual Hooke's law form for the force. The term with $\alpha$ is the nonlinear force.
By introducing $\Dot{x}=p, \Dot{p}=f(x),$ we can further check the system tensors to see whether it is an Hamiltonian system. The matrix $A_2= \left(\begin{matrix}
    \mathbf{0} & \mathbf{P}\\
    \mathbf{T} & \mathbf{0}
\end{matrix}\right)$ has the block structure, where $\mathbf{T},\mathbf{P}$ are symmetric. The $A_2$ is thus Hamiltonian. For the tensor $A_3$, we have $(A_3)_{ijk}=1$ when $j=k=i, i\leq n$, $(A_3)_{ijk}=0$ when $i\leq n$ and $j=k=i$ is not satisfied. For $i\geq n+1$, $(A_3)_{i,:,:}=\left(\begin{matrix}
    \mathbf{\Tilde{x}} & \mathbf{0}\\
    \mathbf{0} & \mathbf{0}
\end{matrix}\right)$, where the matrix $\mathbf{\Tilde{x}}$ can be determined from $f(x)$. It is not hard to check that $A_3$ is also Hamiltonian. Thus, the system is a Hamiltonian system.

However, we emphasize that there exist some polynomial Hamiltonian systems, whose Hamiltonian functions are not in the polynomial form. Let us look at the classical Lotka Volterra model of evolution of species population in ecology \cite{takeuchi1996global,sb2010,duarte1998dynamics,doob1936vito}.
The dynamics read as
\begin{equation}\label{eq:lv}
    \dot{x}_i=x_i\left(r_i+\sum_{j=1}^n a_{i j} x_j\right).
\end{equation}
It is known from \cite{doob1936vito,duarte1998dynamics} that the Lotka-Volterra system \eqref{eq:lv} is conservative if and only if $a_{i i}=0$ and $a_{i j} \neq 0 \Rightarrow a_{i j} a_{j i}<0$, and for every sequence $i_1, i_2, \ldots, i_s$ we have $a_{i_1 i_2} a_{i_2 i_3} \cdots a_{i_s i_1}=(-1)^s a_{i_s i_{s-1}} \cdots a_{i_2 i_1} a_{i_1 i_s}$. The conservative Lotka-Volterra system is proved to be a Hamiltonian system by introducing new coordinates. In addition, the Hamiltonian is not in a polynomial form. There is also an alternative way \cite{sb2010} to show the conservative Lotka-Volterra system \eqref{eq:lv} with an even number of species is a Hamiltonian system without introducing extra coordinates. However, a change of coordinate is needed, which makes the Lotka-Volterra system after the change of coordinate no longer polynomial.

In addition, we provide an example to illustrate the stability results. Consider the system:
\begin{equation}\label{eq:exa}
\begin{aligned}
\Dot{x} &= 4y - y^3, \\
\Dot{y} &= x.
\end{aligned}
\end{equation}
The Hamiltonian for this system is given by
$H = 2y^2 - \frac{1}{3}y^4 - \frac{1}{2}x^2.$
This Hamiltonian can be expressed in the form of \eqref{eq:hami}. For the tensor \( B_2 \), the entries are \( (B_2)_{11} = -\frac{1}{2} \), \( (B_2)_{22} = 2 \), and all other entries are zero. Clearly, \( B_3 \) is an all-zero tensor. For \( B_4 \), the entry \( (B_4)_{2222} = -\frac{1}{3} \), and all other entries are zero.  
The equilibria of \eqref{eq:exa} are \( (0, 0) \), \( (0, 2) \), and \( (0, -2) \). For the origin, since \( \partial^2 H = 2B_2 \) according to Theorem \ref{thm:origin}, and \( B_2 \) has eigenvalues \( -\frac{1}{2} \) and \( 2 \), the origin is a saddle point and thus unstable. For the non-zero equilibria \( x^* \), the second derivative of the Hamiltonian is given by
$\partial^2 H = 12 B_4 (x^*)^2 + 2 B_2.$
For both equilibria \( (0, 2) \) and \( (0, -2) \), \( \partial^2 H \) is negative definite, and therefore these equilibria are stable.  
Of course, for this simple case, one could directly compute the Hessian from the Hamiltonian. However, when using symbolic programming tools like Matlab, the computation becomes time-consuming as the system dimension increases. Instead, we recommend rewriting the system or Hamiltonian in the tensor form \eqref{eq:model1} or \eqref{eq:hami}. The Hessian can then be efficiently computed using tensor-vector multiplications via \eqref{eq:hessian}, which is significantly faster.  
To illustrate this, even for this simple example, running on an M1 Apple Silicon processor using Matlab 2024b, symbolic programming requires approximately \( 0.31 \) seconds to compute the Hessian, while the tensor approach requires only \( 0.00095 \) seconds. As the system dimension grows, this time difference is expected to become even more pronounced.

\section{Further discussion}\label{sec:6}

In this paper, we establish a strong connection between Hamiltonian systems with polynomial Hamiltonian functions and Hamiltonian tensors. We demonstrate that a tensor-based polynomial system is a Hamiltonian system with a polynomial Hamiltonian if and only if all system tensors are Hamiltonian cubical tensors. Furthermore, once the system is identified as a tensor-based polynomial Hamiltonian system, we propose a more efficient method to determine the stability of each equilibrium. Through illustrative examples, we show that our techniques can be applied to efficiently determine whether a polynomial system is Hamiltonian and whether an equilibrium is stable.

For future work, it is worthwhile to explore the standard Hamiltonian equations for mechanical systems with external inputs, which are given by
\begin{equation}\label{eq:port}
\begin{aligned}
\dot{q} &= \frac{\partial H}{\partial p}(q, p), \\
\dot{p} &= -\frac{\partial H}{\partial q}(q, p) + F,
\end{aligned}
\end{equation}
where \( H(q, p) \) is the total energy of the system, \( q = (q_1, \ldots, q_k)^T \) represents the generalized configuration coordinates for a mechanical system with \( k \) degrees of freedom, \( p = (p_1, \ldots, p_k)^T \) is the vector of generalized momenta, and \( F \) is the vector of external generalized forces. Such systems are typically studied using port-Hamiltonian representations. If \( H \) is assumed to be polynomial, further interesting results may be derived by leveraging tensor operations and properties.
On the other hand, as previously mentioned, the properties of Hamiltonian matrices play a significant role in optimal control. It would be another intriguing direction to investigate whether a connection exists between Hamiltonian cubical tensors and certain optimal control problems.

\bibliographystyle{IEEEtran}
\bibliography{bib}

\end{document}